\documentclass[english,aps,preprint,superscriptaddress]{revtex4-1}
\usepackage{lmodern}
\usepackage{lmodern}
\usepackage[T1]{fontenc}
\usepackage[utf8]{inputenc}
\usepackage{amsmath}
\usepackage{amsthm}
\usepackage{amssymb}
\usepackage{graphicx}
\usepackage{esint}

\makeatletter
\theoremstyle{plain}
\newtheorem{thm}{\protect\theoremname}
\theoremstyle{plain}
\newtheorem{lem}[thm]{\protect\lemmaname}
\theoremstyle{plain}
\newtheorem{cor}[thm]{\protect\corollaryname}

\usepackage{amsthm}\usepackage{latexsym}\usepackage{bm}\usepackage{amsfonts}

\usepackage{babel}
\usepackage{hyperref}

\makeatother

\usepackage{babel}
\providecommand{\corollaryname}{Corollary}
\providecommand{\lemmaname}{Lemma}
\providecommand{\theoremname}{Theorem}

\begin{document}
\title{A necessary and sufficient condition for the stability of linear Hamiltonian
systems with periodic coefficients}
\author{Hong Qin}
\email{hongqin@princeton.edu}

\affiliation{Plasma Physics Laboratory, Princeton University, Princeton, NJ 08543,
U.S.A}
\affiliation{School of Physical Sciences, University of Science and Technology
of China, Hefei, 230026, China}
\begin{abstract}
Linear Hamiltonian systems with time-dependent coefficients are of
importance to nonlinear Hamiltonian systems, accelerator physics,
plasma physics, and quantum physics. It is shown that the solution
map of a linear Hamiltonian system with time-dependent coefficients
can be parameterized by an envelope matrix $w(t)$, which has a clear
physical meaning and satisfies a nonlinear envelope matrix equation.
It is proved that a linear Hamiltonian system with periodic coefficients
is stable iff the envelope matrix equation admits a solution with
periodic $\sqrt{w^{\dagger}w}$ and a suitable initial condition.
The mathematical devices utilized in this theoretical development
with significant physical implications are time-dependent canonical
transformations, normal forms for stable symplectic matrices, and
horizontal polar decomposition of symplectic matrices. These tools
systematically decompose the dynamics of linear Hamiltonian systems
with time-dependent coefficients, and are expected to be effective
in other studies as well, such as those on quantum algorithms for
classical Hamiltonian systems. 
\end{abstract}
\maketitle

\section{Introduction and main results \label{sec:Introduction}}

We consider the following $2n$-dimensional linear Hamiltonian system
with periodic time-dependent coefficients,
\begin{align}
\dot{z} & =J\nabla H\thinspace,\label{zdot}\\
J & =\left(\begin{array}{cc}
0 & I_{n}\\
-I_{n} & 0
\end{array}\right)\,,\label{J}\\
H & =\frac{1}{2}z^{\dagger}Az\,,\,\,\,A=\left(\begin{array}{cc}
\kappa\left(t\right) & R\left(t\right)\\
R\left(t\right)^{\dagger} & m^{-1}\left(t\right)
\end{array}\right)\,.\label{H}
\end{align}
Here, $z=\left(x_{1},...,x_{n,}\,p_{1},...,p_{n}\right)^{\dagger}$
are the phase space coordinates, $J$ is the $2n\times2n$ standard
symplectic matrix, and $\kappa(t),$ $R\left(t\right),$ and $m^{-1}\left(t\right)$
are periodic time-dependent $n\times n$ matrices with periodicity
$T$. The matrices $\kappa(t)$ and $m^{-1}\left(t\right)$ are symmetric,
and $m(t)$ is also invertible. Supper script ``$\dagger$'' denotes
matrix transpose. 

We are interested in the stability of the dynamics of system \eqref{zdot}
at $t\rightarrow\infty$. The system is called stable if $\lim_{t\rightarrow\infty}z(t)$
is bounded for all initial conditions $z_{0}=z(t=0).$ The solution
of system \eqref{zdot} can be specified by a solution map,
\begin{equation}
z(t)=M(t)z_{0}\,.\label{zmz0}
\end{equation}
A matrix $M$ is called stable if $\lim_{l\rightarrow\infty}M^{l}$
is bounded, which is equivalent to the condition that $M$ is diagonalizable
with all eigenvalues on the unit circle of the complex plane. In terms
of the one-period solution map $M(T)$, system \eqref{zdot} is stable
iff $M(T)$ is stable. 

The main results of the paper are summarized in the following two
theorems. 
\begin{thm}
\label{thm:sol} The solution map of linear Hamiltonian system (\ref{zdot})
can be written as
\begin{equation}
M(t)=S^{-1}P^{-1}S_{0},\label{M}
\end{equation}
where
\begin{align}
S & \equiv\left(\begin{array}{cc}
w^{\dagger-1} & 0\\
(wR-\dot{w})m & w
\end{array}\right)\,,\label{s}\\
S_{0} & \equiv S(t=0)\,,
\end{align}
 $w(t)$ is a time-dependent $n\times n$ matrix satisfying the envelope
equation
\begin{alignat}{1}
\frac{d}{dt}\left(\frac{dw}{dt}m-wRm\right) & +\frac{dw}{dt}mR^{\dagger}+w\left(\kappa-RmR^{\dagger}\right)-\left(w^{\dagger}wmw^{\dagger}\right)^{-1}=0\,,\label{w}
\end{alignat}
and $P(t)\in Sp(2n,\mathbb{R})\cap O(2n,\mathbb{R})$ is determined
by
\begin{alignat}{1}
\dot{P}=P\left(\begin{array}{cc}
0 & -\mu\\
\mu & 0
\end{array}\right)\,,\label{P}\\
\mu=\left(wmw^{\dagger}\right)^{-1}\,.
\end{alignat}
\end{thm}

\begin{thm}
\label{thm:stable} The linear Hamiltonian system \eqref{zdot} is
stable iff the envelope equation (\ref{w}) admits a solution $w(t)$
such that $S_{0}$ is symplectic and $|w|\equiv\sqrt{w^{\dagger}w}$
is periodic with periodicity $T.$
\end{thm}

The solution map given in Theorem \ref{thm:sol} is constructed using
a time-dependent canonical transformation method. Note that Theorem
\ref{thm:sol} is valid regardless whether the coefficients of the
Hamiltonian are periodic or not. Techniques of normal forms for stable
symplectic matrices and horizontal polar decomposition for symplectic
matrices are developed to prove Theorem \ref{thm:stable}. The results
and techniques leading to Theorem \ref{thm:sol} have been reported
previously in the context of charged particle dynamics in a general
focusing lattice \citep{Qin09-NA,Qin09PoP-NA,Qin09-PRL,Qin10PRL,Chung10,Qin11-056708,Qin13PRL,Qin13PRL2,Chung13,Qin14-044001,Qin15-056702,Chung15,Chung16,Chung16PRL,Chung18}.
These contents are included here for easy reference and self-consistency. 

In Sec.~\ref{sec:Significance}, we will discuss the significance
of the main results and their applications in physics. Section \ref{sec:time-dep}
describes the method of time-dependent canonical transformation, and
Sec.~\ref{sec:Proofs} is devoted to the construction of the solution
map as given in Theorem \ref{thm:sol}. The normal forms of stable
symplectic matrices are presented in Sec.~\ref{sec:Normal}, and
the horizontal polar decomposition of symplectic matrices is introduced
in Sec.~\ref{sec:Horizontal}. The proof of Theorem \ref{thm:stable}
is completed in Sec.~\ref{sec:Proof 2}. 

\section{Significance and applications \label{sec:Significance}}

Linear Hamiltonian systems with time-dependent coefficients \eqref{zdot}
has many important applications in physics and mathematics. In accelerator
physics, it describes charged particle dynamics in a periodic focusing
lattice \citep{Davidson01-all}. The dynamic properties, especially
the stability properties, of the system to a large degree dictate
the designs of beam transport systems and storage rings for modern
accelerators. In the canonical quantization approach for quantum field
theory, Schrödinger's equation in the interaction picture, which is
the starting point of Dyson's expansion, S-matrix and Feynman diagrams,
assumes the form of a linear Hamiltonian system with time-dependent
coefficients. For nonlinear Hamiltonian dynamics, nonlinear periodic
orbit is an important topic. The stability of nonlinear periodic orbits
are described by linear Hamiltonian systems with periodic coefficients.

To illustrate the significance of the main results of this paper,
we look at the implications of Theorems \ref{thm:sol} and \ref{thm:stable}
for the special case of one degree of freedom with $m=1$, for which
Hamilton's equation (\ref{zdot}) reduces to the harmonic oscillator
equation with a periodic spring constant, 
\begin{equation}
\ddot{x}+\kappa(t)x=0\,.\label{harmonic}
\end{equation}
According to Theorem \ref{thm:sol}, the solution of Eq.\,(\ref{harmonic})
is
\begin{align}
\left(\begin{array}{c}
x\\
\dot{x}
\end{array}\right) & =M(t)\left(\begin{array}{c}
x\\
\dot{x}
\end{array}\right)_{0}\,.\\
M\left(t\right) & =S^{-1}\left(\begin{array}{cc}
\cos\phi & \sin\phi\\
-\sin\phi & \cos\phi
\end{array}\right)S_{0}\,,\label{M1}\\
S & =\left(\begin{array}{cc}
w^{-1} & 0\\
-\dot{w} & w
\end{array}\right)\,,\\
\phi\left(t\right) & =\int_{0}^{t}\dfrac{dt}{w^{2}}\,,
\end{align}
where the scalar envelope function $w\left(t\right)$ satisfies the
nonlinear envelope equation
\begin{equation}
\ddot{w}+\kappa\left(t\right)w=w^{-3}\,.\label{w1}
\end{equation}

This solution for Eq.\,(\ref{harmonic}) and the scalar envelope
equation (\ref{w1}) were discovered by Courant and Snyder \citep{Courant58}
in the context of charged particle dynamics in one-dimensional periodic
focusing lattices. The solution map given by Courant and Snyder \citep{Courant58}
is 

\begin{equation}
M\left(t\right)=\left(\begin{array}{cc}
\sqrt{\dfrac{\beta}{\beta_{0}}}\left[\cos\phi+\alpha_{0}\sin\phi\right] & \sqrt{\beta\beta_{0}}\sin\phi\\
-\dfrac{1+\alpha\alpha_{0}}{\sqrt{\beta\beta_{0}}}\sin\phi+\dfrac{\alpha_{0}-\alpha}{\sqrt{\beta\beta_{0}}}\cos\phi & \sqrt{\dfrac{\beta_{0}}{\beta}}\left[\cos\phi-\alpha\sin\phi\right]
\end{array}\right)\,,
\end{equation}
where the $\alpha$ and $\beta$ are time-dependent functions defined
as $\alpha\left(t\right)=-w\dot{w}$ and $\beta\left(t\right)=w^{2}\left(t\right),$
and $\alpha_{0}$ and $\beta_{0}$ are initial conditions at $t=0.$
It equals the $M$ in the three-way splitting form in Eq.\,(\ref{M1}). 

The scalar $w(t)$ is called the envelope because it encapsulates
the slow dynamics of the envelope for the fast oscillation, when the
variation time-scale of $\kappa(t)$ is slow compared with the period
determined by $\kappa(t)$, i.e, 
\begin{equation}
\left|\frac{d\kappa}{\kappa dt}\right|\ll\sqrt{\left|k\right|}\,.
\end{equation}
Since the solution map gives the solution of the dynamics in terms
of $w$, invariants of the dynamics can also be constructed. For example,
the Courant-Synder invariant \citep{Courant58} is 
\begin{equation}
I_{CS}=\dfrac{x^{2}}{w^{2}}+\left(w\dot{x}-\dot{w}x\right)^{2}\,,\label{csi}
\end{equation}
which was re-discovered by Lewis in classical and quantum settings
\citep{Lewis68,Lewis69} . The envelope equation is also known as
the Ermakov-Milne-Pinney equation \citep{Ermakov80,Milne30,Pinney50},
which has been utilized to study 1D time-dependent quantum systems
\citep{Lewis68,Lewis69,Morales88,Monteoliva94} and associated non-adiabatic
Berry phases \citep{Berry85}. Given that harmonic oscillator is the
most important physics problem, it comes as no surprise that Eq.\,(\ref{harmonic})
had been independently examined many times from the same or different
angles over the history \citep{Qin06Sym}. 

When specialized to the system with one degree of freedom (\ref{harmonic}),
Theorem \ref{thm:stable} asserts that the dynamics is stable iff
the envelope equation (\ref{w1}) admits a periodic solution with
periodicity $T$. The theorem emphasizes again the crucial role of
the envelope $w$ in determining the dynamic properties of the system.
The sufficiency is straightforward to establish. If $w(t)$ is periodic
with $w(T)=w(0),$ then the one-period solution map is
\begin{align}
M\left(T\right) & =S_{0}^{-1}\left(\begin{array}{cc}
\cos\phi & \sin\phi\\
-\sin\phi & \cos\phi
\end{array}\right)S_{0}\,.\label{M1m}
\end{align}
 Thus,
\begin{align}
M^{l}\left(T\right) & =S_{0}^{-1}R^{l}(\phi)S_{0}\,,\\
R(\phi) & =\left(\begin{array}{cc}
\cos\phi & \sin\phi\\
-\sin\phi & \cos\phi
\end{array}\right)\,,
\end{align}
and $\lim_{l\rightarrow\infty}M^{l}(T)$ is bounded, since $R^{l}(\phi)=R(l\phi)$
is a rotation in $\mathbb{R}^{2}$. The proof of sufficiency makes
use of the splitting of $M$ in the form Eq.\,(\ref{M1}), which
is the special case of Eq.\,(\ref{M}) for one degree of freedom.

The necessity of Theorem \ref{thm:stable} is difficult to establish,
even for one degree of freedom. Actually, its proof for one degree
of freedom is almost the same as that for higher dimensions, which
is given in Sec.~\ref{sec:Proof 2}. For this purpose, two utilities
are developed. The first pertains to the normal forms of stable symplectic
matrices. It is given by Theorem \ref{thm:normal}, which states that
a stable real symplectic matrix is similar to a direct sum of elements
of $SO(2,\mathbb{R})$ by a real symplectic matrix. Given the fact
that normal forms for symplectic matrices had been investigated by
different authors \citep{Long02,Williamson37,Burgoyne74,Laub74,Wimmer91,Hoermander95,DragtBook},
Theorem \ref{thm:normal} might have been known previously. However,
I have not been able to find it in the literature. A detailed proof
of Theorem \ref{thm:normal} is thus written out in Sec.~\ref{sec:Normal}
for easy reference and completeness. The second utility developed
is the horizontal polar decomposition of symplectic matrices, given
by Theorem \ref{thm:horizontal}. Akin to the situation of normal
forms, Theorem \ref{thm:horizontal} is built upon previous work,
especially that by de Gosson \citep{deGosson06} and Wolf \citep{Wolf04-173}.
The presentation in this paper clarifies some technical issues and
confusions in terminology.

In the present study, Theorems \ref{thm:normal} and \ref{thm:horizontal}
are presented as tools to prove Theorem \ref{thm:stable}. However,
we promote their status to theorems for the importance of their own.
They can be applied to study other problems as well, such as quantum
algorithms for classical physics. 

\section{Method of time-dependent canonical transformation \label{sec:time-dep}}

In this section, we present the method of time-dependent canonical
transformation in preparation for the proof of Theorem \ref{thm:sol}
in the next section. It is necessary to emphasize again that the results
and techniques leading to Theorem \ref{thm:sol} have been reported
previously in the context of charged particle dynamics in a general
focusing lattice \citep{Qin09-NA,Qin09PoP-NA,Qin09-PRL,Qin10PRL,Chung10,Qin11-056708,Qin13PRL,Qin13PRL2,Chung13,Qin14-044001,Qin15-056702,Chung15,Chung16,Chung16PRL,Chung18}.
These contents are included here for easy reference and self-consistency. 

We introduce a time-dependent linear canonical transformation \citep{Leach77}
\begin{equation}
\tilde{z}=S\left(t\right)z\,,\label{zsz}
\end{equation}
 such that in the new coordinate $\tilde{z},$ the transformed Hamiltonian
has the form
\begin{equation}
\tilde{H}=\dfrac{1}{2}\tilde{z}^{\dagger}\tilde{A}\left(t\right)\tilde{z}\,,\label{Hbar}
\end{equation}
 where $\tilde{A}\left(t\right)$ is a targeted symmetric matrix.
The transformation between $z$ and $\tilde{z}$ is canonical, 
\begin{equation}
\dfrac{\partial\tilde{z}_{j}}{\partial z_{k}}J_{kl}\dfrac{\partial\tilde{z}_{j}}{\partial z_{l}}=J_{ij}\,\,\,\textrm{or}\,\,\,SJS^{\dagger}=J\,.\label{zjzj}
\end{equation}
The matrix $S\left(t\right)$ that renders the time-dependent canonical
transformation needs to satisfy a differential equation derived as
follows. With the quadratic form of the Hamiltonian in Eq.\,\eqref{H},
Hamilton's equation\,\eqref{zdot} becomes 
\begin{align}
\dot{z}_{j} & =J_{ij}\dfrac{\partial H}{\partial z_{j}}=\dfrac{1}{2}J_{ij}\left(\delta_{lj}A_{lm}z_{m}+z_{l}A_{lk}\delta_{kj}\right)=J_{ij}A_{jm}z_{m}\,,\label{zdot-ind}
\end{align}
or
\begin{equation}
\dot{z}=JAz\,.\label{zd}
\end{equation}
 Because we require that in $\tilde{z}$ the transformed Hamiltonian
is given by Eq.\,(\ref{Hbar}), the following equation holds as well,
\begin{equation}
\dot{\tilde{z}}=J\tilde{A}\tilde{z}.\label{zdja}
\end{equation}
 Using Eq.\,\eqref{zsz}, we rewrite Eq.\,\eqref{zdja} as 
\begin{equation}
\dot{\tilde{z}}=J\tilde{A}\tilde{z}=J\tilde{A}Sz\,.\label{zbd}
\end{equation}
 Meanwhile, $\dot{\tilde{z}}$ can be directly calculated from Eq.\,(\ref{zsz})
by taking a time-derivative, 
\begin{equation}
\dot{\tilde{z}}=\dot{S}z+S\dot{z}=\left(\dot{S}+SJA\right)z\,.\label{zbd2}
\end{equation}
 Combining Eqs.\,(\ref{zbd}) and (\ref{zbd2}) gives the differential
equation for $S,$
\begin{equation}
\dot{S}=J\tilde{A}S-SJA\,.\label{S}
\end{equation}

\begin{lem}
\label{lem:Seq}The solution $S$ of Eq.\,\eqref{S} is always symplectic,
if $S$ is symplectic at $t=0$. 
\end{lem}

\begin{proof}
We follow Leach \citep{Leach77} and consider the dynamics of the
matrix $K\equiv SJS^{\dagger},$ 
\begin{align}
\dot{K} & =\dot{S}JS^{\dagger}+SJ\dot{S}^{\dagger}\nonumber \\
 & =2\left[\left(J\tilde{A}S-SJA\right)JS^{\dagger}+SJ\left(-S\tilde{A}J+AJS^{\dagger}\right)\right]\nonumber \\
 & =2\left[J\tilde{A}SJS^{\dagger}-SJS^{\dagger}\tilde{A}J\right]=2\left[J\tilde{A}K-K\tilde{A}J\right]\,.\label{Kdot}
\end{align}
 The dynamics of $K$ has a fixed point at $K=J.$ If $S(t=0)$ is
symplectic, \textit{i.e.}, $K\left(t=0\right)=J,$ then $K=J$ for
all $t$, and thus $S$ is symplectic for all $t.$ 

A more geometric proof can be given from the viewpoint of the flow
of $S$ (see Fig.\,\ref{Sp}). Because $\tilde{A}$ is symmetric,
$JJ\tilde{A}-\tilde{A}^{\dagger}JJ=0$ and $J\tilde{A}\in sp\left(2n,\mathbb{R}\right)$,
the Lie algebra of $Sp\left(2n,\mathbb{R}\right).$ If $S\in Sp\left(2n,\mathbb{R}\right)$
at a given $t,$ then $J\tilde{A}S$ is in the tangent space of $Sp\left(2n,\mathbb{R}\right)$
at $S$, \textit{i.e.}, $J\tilde{A}S\in T_{S}SP\left(2n,\mathbb{R}\right).$
This can be seen by examining the Lie group right action 
\begin{equation}
S:\,a\mapsto aS
\end{equation}
for any $a$ in $Sp\left(2n,\mathbb{R}\right),$ and the associated
tangent map 
\begin{equation}
T_{S}:\ T_{a}Sp\left(2n,\mathbb{R}\right)\rightarrow T_{aS}Sp\left(2n,\mathbb{R}\right).
\end{equation}
It is evident that $J\tilde{A}S$ is the image of the Lie algebra
element $J\tilde{A}$ under the tangential map $T_{S},$ which means
that $J\tilde{A}S$ is a ``vector”\ tangential to the space of $Sp\left(2n,\mathbb{R}\right)$
at $S.$ The same argument applies to $SJA$ as well. Consequently,
the right hand side of Eq.\,(\ref{S}) is a vector field on $Sp\left(2n,\mathbb{R}\right)$.
The $S$ dynamics will stay on the space of $Sp\left(2n,\mathbb{R}\right)$,
if it does at $t=0.$
\end{proof}
\begin{figure}[ptb]
\begin{centering}
\includegraphics[width=3in]{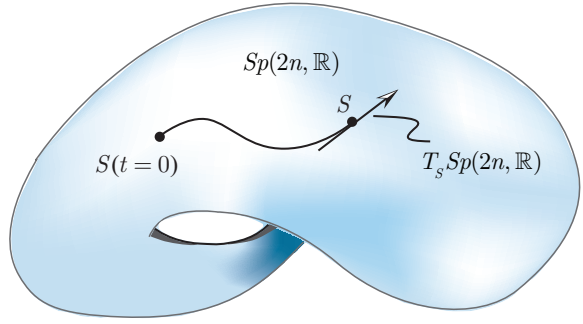} 
\par\end{centering}
\caption{The space of the symplectic group $Sp\left(2n,\mathbb{R}\right)$
and the flow of $S$ on $Sp\left(2n,\mathbb{R}\right).$ At any given
time, $J\tilde{A}S-SJA$ is tangential to $Sp\left(2n,\mathbb{R}\right),$
and the flow of $S$ according to Eq.\,\eqref{S} is always on $Sp\left(2n,\mathbb{R}\right).$ }

\label{Sp} 
\end{figure}

\section{Proof of Theorem 1 \label{sec:Proofs}}

We now apply the technique developed in Sec. \ref{sec:time-dep} to
prove Theorem \ref{thm:sol}. Our goal is to find a new coordinate
system where the transformed Hamiltonian vanishes. This idea is identical
to that in Hamilton-Jacobi theory. The goal is accomplished in two
steps. First, we seek a coordinate transformation $\tilde{z}=Sz$
such that, in the $\tilde{z}$ coordinates, the Hamiltonian assumes
the form 
\begin{equation}
\bar{H}=\frac{1}{2}\tilde{z}^{\dagger}\tilde{A}\tilde{z}\,,\,\,\tilde{A}=\left(\begin{array}{cc}
\mu(t) & 0\\
0 & \mu(t)
\end{array}\right)\,,
\end{equation}
where $\mu(t)$ is a $n\times n$ matrix to be determined. Let the
$2\times2$ block form of $S$ is 
\begin{equation}
S=\left(\begin{array}{cc}
S_{1} & S_{2}\\
S_{3} & S_{4}
\end{array}\right),
\end{equation}
and split Eq.\,\eqref{S} into four $n\times n$ matrix equations,
\begin{align}
\dot{S}_{1} & =\mu S_{3}-S_{1}R^{\dagger}+S_{2}\kappa\,\,,\label{eq:S1-1}\\
\dot{S}_{2} & =\mu S_{4}-S_{1}m^{-1}+S_{2}R\,,\label{eq:S2-1}\\
\dot{S}_{3} & =-\mu S_{1}-S_{3}R^{\dagger}+S_{4}\kappa\,,\\
\dot{S}_{4} & =-\mu S_{2}-S_{3}m^{-1}+S_{4}R\,.\label{eq:S4-1}
\end{align}
Including $\mu(t),$ there are five $n\times n$ matrices unknown.
The extra freedom is introduced by the to-be-determined $\mu(t)$.
We choose $S_{2}\equiv0$ to remove the freedom, and rename $S_{4}$
to be $w,$ i.e., $w\equiv S_{4}$. Equations \eqref{eq:S1-1}-\eqref{eq:S4-1}
become
\begin{align}
\dot{S}_{1} & =\mu S_{3}-S_{1}R^{\dagger}\,,\label{S1}\\
S_{1} & =\mu wm\,,\label{S2}\\
\dot{S}_{3} & =-\mu S_{1}-S_{3}R^{\dagger}+w\kappa\,,\label{S3}\\
S_{3} & =-\dot{w}m+wRm,\label{S4}
\end{align}
for matrices $S_{1}$, $S_{3}$, $w$ and $\mu$. Because $(S_{1},S_{2}=0,S_{3},S_{4}=w)$
describes a curve in $Sp\left(2n,\mathbb{R}\right),$ symplectic condition
$S_{1}S_{4}^{\dagger}-S_{2}S_{3}^{\dagger}=I$ holds, which implies
\begin{equation}
S_{1}=w^{\dagger-1}\,.\label{eq:S1}
\end{equation}
From Eq.\,\eqref{S2}, we have
\begin{equation}
\mu=\left(wmw^{\dagger}\right)^{-1}.\label{mu}
\end{equation}
Equation \eqref{S1} is equivalent to another symplectic condition
$S_{3}S_{4}^{\dagger}=S_{4}S_{3}^{\dagger}.$ Substituting Eqs.\,\eqref{S4}-\eqref{mu}
into Eq.\,\eqref{S3}, we obtain the following matrix differential
equation for the envelope matrix $w,$
\begin{equation}
\frac{d}{dt}\left(\frac{dw}{dt}m-wRm\right)+\frac{dw}{dt}mR^{\dagger}+w\left(\kappa-RmR^{\dagger}\right)-\left(w^{\dagger}wmw^{\dagger}\right)^{-1}=0\,.\label{w-1}
\end{equation}
This is the desired envelope equation in Theorem \ref{thm:sol}. 

Once $w$ is solved for from the envelope equation, we can determine
$S_{1}$ from Eq.\,\eqref{eq:S1} and $S_{3}$ from Eq.\,\eqref{S4}.
In terms of the envelope matrix $w,$ the symplectic transformation
$S$ and its inverse are given by
\begin{align}
S & =\left(\begin{array}{cc}
w^{\dagger-1} & 0\\
(wR-\dot{w})m & w
\end{array}\right)\,,\label{eq:S=00003D}\\
S^{-1} & =\left(\begin{array}{cc}
w^{\dagger} & 0\\
\left(w^{-1}\dot{w}-R\right)mw^{\dagger} & w^{-1}
\end{array}\right)\,.\label{eq:S-1=00003D}
\end{align}

The second step is to use another coordinate transformation $\tilde{\tilde{z}}=P(t)\tilde{z}\,$
to transform $\tilde{H}$ into a vanishing Hamiltonian $\tilde{\tilde{H}}\equiv0$
at all time, thereby rendering the dynamics trivial in the new coordinates.
The determining equation for the transformation $P(t)$ is 
\begin{equation}
\dot{P}=-PJ\tilde{A}=P\left(\begin{array}{cc}
0 & -\mu\\
\mu & 0
\end{array}\right)\,.\label{P-1}
\end{equation}
According to Lemma \ref{lem:Seq}, the $P$ matrix satisfying Eq.\,\eqref{P-1}
is symplectic because $J\tilde{A}\in sp(2n,\mathbb{R}).$ From $\mu=\mu^{\dagger}$,
we know that $J\tilde{A}$ is also antisymmetric, i.e., $J\tilde{A}\in so(2n,\mathbb{R})$.
Thus $J\tilde{A}\in sp(2n,\mathbb{R})\cap o(2n,\mathbb{R})$, and
$P(t)$ is a curve in the group of $2n$-dimensional symplectic rotations,
i.e., $P(t)\in Sp(2n,\mathbb{R})\cap O(2n,\mathbb{R})\backsimeq U(n)$,
provided $P(t)$ starts from the group at $t=0.$ We call $P(t)$
the phase advance, an appropriate descriptor in light of the fact
that $P(t)$ is a symplectic rotation. The Lie algebra element (infinitesimal
generator) $-J\tilde{A}=\left(\begin{array}{cc}
0 & -\mu\\
\mu & 0
\end{array}\right)$ is the phase advance rate, and it is determined by the envelope matrix
through Eq.\,\eqref{mu}. As an element in $Sp(2n,\mathbb{R})\cap O(2n,\mathbb{R})\backsimeq U(n)$,
$P$ and its inverse must have the forms
\begin{align}
P & =\left(\begin{array}{cc}
P_{1} & P_{2}\\
-P_{2} & P_{1}
\end{array}\right)\,,\label{eq:P=00003D}\\
P^{-1} & =P^{\dagger}=\left(\begin{array}{cc}
P_{1}^{\dagger} & -P_{2}^{\dagger}\\
P_{2}^{\dagger} & P_{1}^{\dagger}
\end{array}\right)\,.\,\,\label{eq:P-1=00003D}
\end{align}

Combining the two time-dependent canonical transformations, we have
\begin{equation}
\tilde{\tilde{z}}=G(t)z=P(t)S(t)z\,.\label{G}
\end{equation}
In the $\tilde{\tilde{z}}$ coordinates, because $\tilde{\tilde{H}}\equiv0$,
the dynamics is trivial, i.e., $\tilde{\tilde{z}}=const.$ This enables
us to construct the solution map \eqref{zmz0} as
\begin{equation}
M(t)=S^{-1}P^{-1}P_{0}S_{0}=\left(\begin{array}{cc}
w^{\dagger} & 0\\
\left(w^{-1}\dot{w}-R\right)mw^{\dagger} & w^{-1}
\end{array}\right)P^{\dagger}\left(\begin{array}{cc}
w^{-\dagger} & 0\\
(wR-\dot{w})m & w
\end{array}\right)_{0},\label{Md}
\end{equation}
where subscript ``0'' denotes initial conditions at $t=0$, and
$P_{0}$ is taken to be $I$ without loss of generality. 

This completes the proof of Theorem \eqref{thm:sol}.

The time-dependent canonical transformation can also be used to construct
invariants of the dynamics. For any constant $2n\times2n$ positive-definite
matrix $\xi,$ the quantity 
\begin{equation}
I_{\xi}=z^{\dagger}S^{\dagger}P^{\dagger}\xi PSz\label{Ixi}
\end{equation}
is a constant of motion, since $\tilde{\tilde{z}}=PSz$ is a constant
of motion. For the special case of $\xi=I,$ the phase advance $P$
in Eq.\,\eqref{Ixi} drops out, and 
\begin{equation}
I_{CS}\equiv z^{\dagger}S^{\dagger}Sz=z^{\dagger}\left(\begin{array}{cc}
\gamma & \alpha\\
\alpha^{T} & \beta
\end{array}\right)z\,,\label{eq:ICS2}
\end{equation}
where $\alpha,$ $\beta,$ and $\gamma$ are $2\times2$ matrices
defined by
\begin{align}
\alpha & \equiv w^{\dagger}S_{3}\,,\label{eq:alphag}\\
\beta & \equiv w^{\dagger}w\,,\label{eq:betag}\\
\gamma & \equiv S_{3}^{\dagger}S_{3}+w^{-1}w^{-\dagger}\,.\label{eq:gammag}
\end{align}
Here, we use $I_{CS}$ to denote this special invariant because it
is the invariant that generalizes the Courant-Snyder invariant \citep{Courant58}
(or Lewis invariant \citep{Lewis68,Lewis69}) for one degree of freedom
in Eq.\,\eqref{csi}. 

\section{Normal forms for stable symplectic matrices \label{sec:Normal}}

In this section, normal forms for stable symplectic matrices are developed.
We first list necessary definitions and facts related to the eigenvalues,
eigenvector spaces, and root-vector spaces of symplectic matrices.
Lemmas \ref{lem:Sspectrum}-\ref{lem:Kpp} are given without proofs,
which can be found in the textbooks by Long \citep{Long02}, Ekeland
\citep{Ekeland90}, and Yakubovich and Starzhinskii \citep{Yakubovich75}.

The set of eigenvalues of a matrix $M$ is denoted by
\begin{equation}
\sigma(M)=\left\{ \lambda|\det(I\lambda-M)=0\right\} ,
\end{equation}
 and the unit circle on the complex plane $\mathbb{C}$ is denoted
by $U.$
\begin{lem}
\label{lem:Sspectrum} For a symplectic matrix $M\in Sp(2n,\mathbb{R}),$
its eigenvalue space $\sigma(M)$ is symmetric with respect to the
real axis and the unit circle $U$, i.e., if $\lambda_{0}\in\sigma(M),$
then $\bar{\lambda}_{0},$$\lambda_{0}^{-1}\in\sigma(M).$
\end{lem}

The symmetry of $\sigma(M)$ with respect to the real axis is not
specific to symplectic matrices. It is true for all real matrices.
For a $\lambda\in\sigma(M)$, denote the eigenvector space by $V_{\lambda}(M)$
and the root-vector space by $E_{\lambda}(M)$,

\begin{align}
V_{\lambda}(M) & \equiv\textrm{ker}(M-\lambda I)\subset\mathbb{C}^{2n},\\
E_{\lambda}(M) & \equiv\cup_{k\geq1}\textrm{ker}(M-\lambda I)^{k}\subset\mathbb{C}^{2n}.
\end{align}
The dimension of $V_{\lambda}(M)$ is the geometric multiplicity of
$\lambda,$ denoted by $Mult_{G\lambda}(M)$, and the dimension of
$E_{\lambda}(M)$ is the algebraic multiplicity of $\lambda,$ denoted
by $Mult_{A\lambda}(M)$. A subspace $V\subset\mathbb{C}^{2n}$ is
an invariant subspace if $MV\subset\mathbb{C}^{2n}$. An invariant
subspace is irreducible if it is not a direct sum of two non-trivial
invariant subspaces. For every $\lambda$, the root-vector space $E_{\lambda}(M)$
is an invariant subspace, and a direct sum of irreducible subspaces.
On the other hand, every irreducible subspace of $M$ is contained
in one of the root-vector spaces. An eigenvalue $\lambda$ is simple
if $Mult_{G\lambda}(M)=$$Mult_{A\lambda}(M)=1$. An eigenvalue $\lambda$
is semi-simple, if $E_{\lambda}(M)$ is a direct sum of one-dimensional
irreducible invariant subspaces only, which is equivalent to that
all elementary divisors of $\lambda$ are simple. When $\lambda$
is semi-simple, the eigenvector space $V_{\lambda}(M)$ is at its
maximum dimension. It has fulfilled its obligation to provide enough
eigenvectors for the purpose of diagonalizing $M$, even though $Mult_{G\lambda}(M)=$$Mult_{A\lambda}(M)$
could be larger than $1$. If $M$ is not diagonalizable, other eigenvalues
have to be blamed. 

For any two vectors $\psi$ and $\phi$ in $\mathbb{C}^{2n},$ the
Krein product is 
\begin{equation}
\left\langle \psi,\phi\right\rangle _{G}\equiv\left\langle G\psi,\phi\right\rangle =-\phi^{*}iJ\psi,
\end{equation}
where $G=-iJ$ and $\phi^{*}=\bar{\phi}^{\dagger}$ is the Hermitian
transpose of $\phi.$ Obviously, 
\begin{equation}
\left\langle \psi,\phi\right\rangle _{G}=\overline{\left\langle \phi,\psi\right\rangle _{G}}\,.
\end{equation}
For a vector $\psi$ in $\mathbb{C}^{2n},$ it Krein amplitude is
defined to be $\left\langle \psi,\psi\right\rangle _{G}$, the sign
of which is the Krein signature of $\psi.$ 

Two vectors $\psi$ and $\phi$ in $\mathbb{C}^{2n}$ are G-orthogonal
if $\left\langle \psi,\phi\right\rangle _{G}=0.$ Two subspaces $V_{1}$
and $V_{1}$ are G-orthogonal if $\left\langle \psi_{1},\psi_{2}\right\rangle _{G}=0$
for any $\psi_{1}\in V_{1}$ and $\psi_{2}\in V_{2}.$ A subspace
is G-isotropic if it is G-orthogonal to itself. 

Krein amplitude has a clear physical meaning. The one-period solution
map $M(T)$ according to Eq.\,(\ref{Md}) can be expressed as
\begin{equation}
M(T)=\exp(J\hat{A}T)\,,
\end{equation}
where $\hat{A}$ is an effective matrix representing the averaged
effect by $A$ in one period. In physics, an eigenvector $\psi\in V_{\lambda}(M)$
is known as an eigenmode, and the jargon of eigen-frequency $\omega$,
defined by $\exp(-i\omega T)=\lambda$, is preferred. An eigenmode
$\psi$ of $M(T)$ is also an eigenmode of $J\hat{A}$,
\begin{equation}
J\hat{A}\psi=-i\omega\psi\,.
\end{equation}
The Krein amplitude of $\psi$ is 
\begin{equation}
\left\langle \psi,\psi\right\rangle _{G}=-i\psi^{*}J\psi=-\frac{\psi^{*}\hat{A}\psi}{\omega}=-2\frac{\hat{H}}{\omega}\,,
\end{equation}
which is proportional to the negative of the action of the eigenmode,
i.e., the time integral of the effective energy over one period. Because
the ratio between $\left\langle \psi,\psi\right\rangle _{G}$ and
$\hat{H}/\omega$ is an unimportant constant, we can identify the
Krein amplitude with the action of the eigenmode \citep{Zhang16GH,Zhang1711,Zhang18}. 

The following lemma gives a necessary condition for an eigenvalue
on the unit circle to be non-semi-simple. 
\begin{lem}
For a $\lambda\in\sigma(M)\cap U$ with at least one multiple elementary
divisor, there exists $\psi\in E_{\lambda}(M)$ such that $\left\langle \psi,\psi\right\rangle _{G}=0$
. 
\end{lem}

\begin{lem}
\label{lem:VGO}For two eigenvalues $\lambda$ and $\mu$, the eigenvector
spaces $V_{\lambda}$ and $V_{\mu}$ are G-orthogonal if $\lambda\bar{\mu}\neq1.$ 
\end{lem}

It is remarkable that the G-orthogonality can be established for the
root-vector spaces as well.
\begin{lem}
\label{lem:EGO}For two eigenvalues $\lambda$ and $\mu$, the root-vector
spaces $E_{\lambda}$ and $E_{\mu}$ are G-orthogonal if $\lambda\bar{\mu}\neq1.$ 
\end{lem}

Lemma \ref{lem:EGO} implies for any $\lambda\in\sigma(M)\cap U,$
$E_{\lambda}$ is G-orthogonal to $E_{\mu}$ when $\mu\neq\lambda.$
As a consequence, $\mathbb{C}^{2n}$ has the following G-orthogonal
decomposition \citep{Long02}
\begin{align}
\mathbb{C}^{2n} & =\left(\oplus_{\lambda\in\sigma(M)\cap U}E_{\lambda}(M)\right)\oplus F(M)\,,
\end{align}
where 
\begin{equation}
F(M)=\oplus_{\lambda\in\sigma(M)\backslash U}E_{\lambda}(M)\,
\end{equation}
is the direct sum of the root-vector spaces for all eigenvalues off
the unit circle. Note that root-vector spaces for different eigenvalues
off the unit circle are not G-orthogonal in general. 
\begin{lem}
\label{lem:GE} For a $\lambda\in\sigma(M)\cap U$, $E_{\lambda}(M)$
is an invariant subspace of $G=-iJ.$ Furthermore, the restriction
of $G$ on $E_{\lambda}(M)$ is non-degenerate, i.e., for $x$ and
$y$ in $E_{\lambda}(M)$, $\left\langle x,y\right\rangle _{G}=0$
for all $y$ implies that $x=0.$ 
\end{lem}

As an operator on $\mathbb{C}^{2n}$, $G=-iJ$ is Hermitian, i.e.,
$\left(-iJ\right)^{*}=-iJ$. It is also non-degenerate because $\text{det}(G)\neq0$.
Since a Hermitian matrix is diagonalizable and its eigenvalues are
all real, there exists a G-orthogonal basis for $\mathbb{C}^{2n}$. 

By Lemma \ref{lem:GE}, the restriction of $G$ on $E_{\lambda}(M)$
for $\lambda\in\sigma(M)\cap U$ is a non-degenerate Hermitian operator,
which implies $G|_{E_{\lambda}(M)}$ is diagonalizable with non-zero
eigenvalues. Similar to the situation for $\mathbb{C}^{2n}$, the
space of $E_{\lambda}(M)$ admits a G-orthogonal base. Let the dimension
of of $E_{\lambda}(M)$ is $m$, and $p$ and $q$ are the numbers
of positive and negative eigenvalues of $G|_{E_{\lambda}(M)}$ respectively.
The pair $(p,q)$ is known as the Krein type of $\lambda.$ Obviously,
$p+q=m.$ If $q=0$, $\lambda$ is Krein-positive, and if $p=0$,
$\lambda$ is Krein-negative. An eigenvalue $\lambda$ is Krein-definite,
if it is either Krein-positive or Krein-negative. Otherwise, $\lambda$
is Krein-indefinite or of mixed type.
\begin{lem}
For a $\lambda\in\sigma(M)\cap U$ and any $P\in Sp(2n,\mathbb{R}),$
the geometric and algebraic multiplicities and Krein type of $\lambda$
are identical for $P^{-1}MP$ and $M.$ 
\end{lem}

\begin{lem}
\label{lem:Kpp}

(a) For any $x\in\mathbb{C}^{2n}$, $\left\langle Gx,x\right\rangle =-\left\langle G\bar{x},\bar{x}\right\rangle .$ 

(b) If $\lambda\in\sigma(M)\cap U$ has Krein type $(p,q)$, then
$\bar{\lambda}$ has Krein type $(q,p)$. In particular, if $1$ or
$-1$ is an eigenvalue of $M$, its Krein type is $(p,p)$ for some
$p\in\mathbb{N}$.
\end{lem}

The $\diamond$-product of two square matrices introduced by Long
\citep{Long02} is an indispensable tool in the manipulations of symplectic
matrices. Let $M_{1}$ be a $2i\times2i$ matrix and $M_{2}$ a $2j\times2j$
matrix in the square block form, 
\begin{equation}
M_{1}=\left(\begin{array}{cc}
A_{1} & B_{1}\\
C_{1} & D_{1}
\end{array}\right)\,,\,\,\,\,M_{2}=\left(\begin{array}{cc}
A_{2} & B_{2}\\
C_{2} & D_{2}
\end{array}\right)\,.
\end{equation}
The $\diamond$-product of $M_{1}$ and $M_{2}$ is defined to be
a $2(i+j)\times2(i+j)$ matrix as 
\begin{equation}
M_{1}\diamond M_{2}=\left(\begin{array}{cccc}
A_{1} & 0 & B_{1} & 0\\
0 & A_{2} & 0 & B_{2}\\
C_{1} & 0 & D_{1} & 0\\
0 & C_{2} & 0 & D_{2}
\end{array}\right)\,.
\end{equation}
The $\diamond$-product defined here is compatible with the standard
symplectic matrix defined in Eq.\,(\ref{J}). The $\diamond$-product
of two symplectic matrices is symplectic \citep{Long02}. It can be
viewed as a direct sum of two matrix vectors to form a matrix vector
in higher dimension. The following theorem is the main result of this
section, which establishes the normal form for a stable symplectic
matrix. 
\begin{thm}
\label{thm:normal}For a symplectic matrix $M\in Sp(2n,\mathbb{R})$,
it is stable iff it is similar to a direct sum of elements in $SO(2,\mathbb{R})$
by a symplectic matrix, i.e., there exist a $N=R(\theta_{1})\diamond R(\theta_{2})...\diamond R(\theta_{n})$
and a $F\in Sp(2n,\mathbb{R})$ such that 
\begin{equation}
M=FNF^{-1},\label{eq:Mnormal}
\end{equation}
 where $\theta_{j}\in\mathbb{R}$ and 
\begin{equation}
R(\theta_{j})=\left(\begin{array}{cc}
\cos\theta_{j} & \sin\theta_{j}\\
-\sin\theta_{j} & \cos\theta_{j}
\end{array}\right)\in SO(2,\mathbb{R}).
\end{equation}

\end{thm}

\begin{proof}
Recall that stability of $M$ means that $M$ is diagonalizable and
all eigenvalues of $M$ locate on the unit circle $U$ of the complex
plane. It is straightforward to verify that $N=R(\theta_{1})\diamond R(\theta_{2})...\diamond R(\theta_{n})\in Sp(2n,\mathbb{R})\cap O(2n,\mathbb{R})\backsimeq U(n).$
Since a unitary matrix is stable, the sufficiency is obvious. 

For necessity, assume that $M\in Sp(2n,\mathbb{R})$ is stable. First,
let's construct a G-orthonormal basis $(\psi_{l},\psi_{-l})\,,\,l=1,2,...,n,$
for $\mathbb{C}^{2n}$, which by definition satisfies the following
conditions for $1\leq l,m\leq n$,
\begin{align}
M\psi_{l} & =\lambda_{l}\psi_{l}\,,\label{eq:base1}\\
\psi_{-l} & =\bar{\psi}_{l}\,,\\
\left\langle G\psi_{l},\psi_{m}\right\rangle  & =\delta_{lm\,,}\label{eq:base3}\\
\left\langle G\psi_{-l},\psi_{-m}\right\rangle  & =-\delta_{lm\,,}\\
\left\langle G\psi_{l},\psi_{-m}\right\rangle  & =0\,,\label{eq:base5}
\end{align}
where $\left(\lambda_{l},\psi_{l}\right)$ is a pair of eigenvalue
and eigenvector, so is $\left(\bar{\lambda}_{l},\psi_{-l}\right)$.
Thus this G-orthonormal basis $(\psi_{l},\psi_{-l})\,\,(l=1,2,...,n)$
for $\mathbb{C}^{2n}$ consists of eigenvectors of $M$, with appropriately
chosen labels. There are in total $2n$ eigenvectors. Because $M$
is stable, every eigenvalue of $M$ is on the unit circle $U$ and
is either simple or semi-simple, and for every $\lambda\in\sigma(M)$,
the eigenvector space $V_{\lambda}(M)$ is identical to the root-vector
space $E_{\lambda}(M).$ The simple and semi-simple cases need to
be treated differently. 

Case a). For a simple eigenvalue $\lambda$ on $U$, $V_{\lambda}(M)=E_{\lambda}(M)$
is one-dimensional. Denote the eigenvector by $\psi_{\lambda}$. According
to Lemma \ref{lem:EGO}, $\psi_{\lambda}$ is G-orthogonal to any
other eigenvector or root-vector of $M.$ By Lemma \ref{lem:GE},
$\left\langle \psi_{\lambda},\psi_{\lambda}\right\rangle _{G}\neq0.$
Without losing generality, we can let $\left\langle \psi_{\lambda},\psi_{\lambda}\right\rangle _{G}=1$
or $-1$. Now, for every simple $\lambda$ on $U$, $\bar{\lambda}$
is in $\sigma(M)\cap U$ and $\bar{\lambda}\neq\lambda$. Furthermore,
$\bar{\lambda}$ is also simple with opposition Krein signature, and
the corresponding eigenvector $\psi_{\bar{\lambda}}$ is G-orthogonal
to all other eigenvectors or root-vectors. All the eigenvectors corresponding
to simple eigenvalues pair up nicely in this manner. It is natural
to denote them by $\left(\psi_{l},\psi_{-l}\right),\,1\leq l\leq n_{s}$,
where $2n_{s}$ is the total number of simple eigenvalues and $l$
is the index for the simple eigenvalues. The pair of eigenvectors
$\left(\psi_{l},\psi_{-l}\right)$ correspond to the pair of eigenvalues
$(\lambda_{l},\bar{\lambda}_{l}),$ where $\psi_{l}$ is the eigenvector
with positive Krein signature, and $\psi_{-l}$ with the negative
Krein signature. Therefore, $\left(\psi_{l},\psi_{-l}\right),\,1\leq l\leq n_{s}$,
form a G-orthonormal basis satisfying Eqs.\,(\ref{eq:base1})-(\ref{eq:base5}),
for the subspace spanned by the eigenvectors of simple eigenvalues
on the unit circle $U.$ 

Case b). For a semi-simple eigenvalue $\mu$ on $U$, the subspace
$V_{\mu}(M)=E_{\mu}(M)$ is two-dimensional or higher. Since $G$
is Hermitian on $E_{\mu}(M),$ it can be diagonalized using a proper
basis $\psi_{\mu}^{i},$ $i=1,...,m_{\mu}$, on $E_{\mu}(M),$ where
$m_{\mu}$ is the multiplicity of $\mu.$ By Lemma \ref{lem:GE},
the restriction of $G$ on $E_{\mu}(M)$ is non-degenerate. The eigenvalues
of $G$ is either positive or negative. Let the Krein type of $\mu$
is $(p,q)$ with $p+q=m_{\mu}.$ Lemma \ref{lem:Kpp} asserts that
the Krein type of $\bar{\mu}$ is $(q,p),$ with $\psi_{\bar{\mu}}^{i}=\bar{\psi}_{\mu}^{i}.$ 

Case b1). When $\mu\neq\pm1,$ $\mu\neq\bar{\mu}$ and $E_{\mu}(M)$
is G-orthogonal to $E_{\bar{\mu}}(M).$ We can thus pair up $\psi_{\mu}^{i}$
and $\psi_{\bar{\mu}}^{i}$ for $i=1,...,m_{\mu}$ to form a G-orthonormal
basis satisfying Eqs.\,(\ref{eq:base1})-(\ref{eq:base5}) for the
subspace $E_{\mu}(M)\oplus E_{\bar{\mu}}(M).$

Case b2). When $\mu=\pm1,$ $E_{\mu}(M)=E_{\bar{\mu}}(M)$ and $p=q=m_{\mu}/2$.
In this case, we can construct a G-orthonormal base of $E_{\mu}(M)$
in the following way. Let $\psi_{\mu}^{i},$ $i=1,...,p,$ be eigenvectors
with positive Krein signatures. Then according to Lemma \ref{lem:Kpp},
$\bar{\psi}_{\mu}^{i},$ $i=1,...,p,$ will have negative Krein signatures.
Note also that $\left(\psi_{\mu}^{i},\bar{\psi}_{\mu}^{i}\right),$
$i=1,...,p$, diagonalize $G$ on $E_{\mu}(M).$ Thus, the pairs $\left(\psi_{\mu}^{i},\bar{\psi}_{\mu}^{i}\right),$
$i=1,...,p$, make up a G-orthonormal base satisfying Eqs.\,(\ref{eq:base1})-(\ref{eq:base5})
for the subspace $E_{\mu}(M)=E_{\bar{\mu}}(M)$. 

Assembling the eigenvector pairs from Cases a), b1), b2) in the order
presented, we obtain the G-orthonormal base $(\psi_{l},\psi_{-l}),\,l=1,2,...,n,$
for $\mathbb{C}^{2n}$ satisfying Eqs.\,(\ref{eq:base1})-(\ref{eq:base5}).
Note that the eigenvalues $\lambda_{l}$ $(1\leq l\leq n_{s})$ are
simple and distinct, whereas the eigenvalues $\lambda_{l}$ $(n_{s}+1\leq l\leq n)$
are semi-simple, and each eigenvalue may appear several times in the
sequence.

Now, we explicitly build the normal form of Eq.\,(\ref{eq:Mnormal})
using the G-orthonormal base $(\psi_{l},\psi_{-l}),\,l=1,2,...,n$.
Let 
\begin{equation}
\psi_{l}=\xi_{l}+\sqrt{-1}\eta_{l}\,,\,\,\lambda_{l}=c_{l}+\sqrt{-1}s_{l}\,,
\end{equation}
where $\xi_{l}$, $\eta_{l}\in\mathbb{R}^{2n}$, $c_{l}=\cos\theta_{l},$
$s_{l}=\sin\theta_{l},$ and $0\leq\theta_{l}<2\pi$. In terms of
this set of real variables, $M\psi_{l}=\lambda_{l}\psi_{l}$ is 
\begin{align}
M\xi_{l} & =c_{l}\xi_{l}-s_{l}\eta_{l}\,,\label{eq:Mxi}\\
M\eta_{l} & =c_{l}\eta_{l}+s_{l}\xi_{l}\,.\label{eq:Met}
\end{align}
The G-orthonormal conditions (\ref{eq:base3})-(\ref{eq:base5}) are
equivalent to
\begin{align}
\left\langle J\xi_{l},\xi_{m}\right\rangle  & =0\,,\label{eq:jxi1}\\
\left\langle J\eta_{l},\eta_{m}\right\rangle  & =0\,,\\
\left\langle J\xi_{l},\eta_{m}\right\rangle  & =\eta_{m}^{\dagger}J\xi_{l}=-\delta_{lm}/2\,,\\
\left\langle J\eta_{l},\xi_{m}\right\rangle  & =\xi_{m}^{\dagger}J\eta_{l}=\delta_{lm}/2\,.\label{eq:jet}
\end{align}
We now prove that the matrix $F$ in Eq.\,(\ref{eq:Mnormal}) in
given by
\begin{equation}
F=\sqrt{2}\left(\xi_{1},\xi_{2},...,\xi_{n},\eta_{1},\eta_{2},...,\eta_{n}\right)\,.
\end{equation}
The fact that $F\in Sp(2n,\mathbb{R})$ is shown by direct calculation,
\begin{align}
F^{\dagger}JF & =2\left(\begin{array}{c}
\xi_{1}^{\dagger}\\
\vdots\\
\xi_{n}^{\dagger}\\
\eta_{1}^{\dagger}\\
\vdots\\
\eta_{n}^{\dagger}
\end{array}\right)J\left(\xi_{1},\xi_{2},...,\xi_{n},\eta_{1},\eta_{2},...,\eta_{n}\right)\nonumber \\
 & =2\left(\begin{array}{cccccc}
\xi_{1}^{\dagger}J\xi_{1} & \ldots & \xi_{1}^{\dagger}J\xi_{n} & \xi_{1}^{\dagger}J\eta_{1} & \cdots & \xi_{1}^{\dagger}J\eta_{n}\\
\vdots & \ddots & \vdots & \vdots & \ddots & \vdots\\
\xi_{1}^{\dagger}J\xi_{1} & \ldots & \xi_{1}^{\dagger}J\xi_{n} & \xi_{1}^{\dagger}J\eta_{1} & \cdots & \xi_{1}^{\dagger}J\eta_{n}\\
\eta_{1}^{\dagger}J\xi_{1} & \ldots & \eta_{1}^{\dagger}J\xi_{n} & \eta_{1}^{\dagger}J\eta_{1} & \cdots & \eta_{1}^{\dagger}J\eta_{n}\\
\vdots & \ddots & \vdots & \vdots & \ddots & \vdots\\
\eta_{n}^{\dagger}J\xi_{1} & \ldots & \eta_{n}^{\dagger}J\xi_{n} & \eta_{n}^{\dagger}J\eta_{1} & \cdots & \eta_{n}^{\dagger}J\eta_{n}
\end{array}\right)\nonumber \\
 & =J\,,
\end{align}
where Eqs.\,(\ref{eq:jxi1})-(\ref{eq:jet}) are used in the last
equal sign. The last step is to show that $N=F^{-1}MF=-JF^{\dagger}JMF$
is of the form $R(\theta_{1})\diamond R(\theta_{2})...\diamond R(\theta_{n})$,
again by direct calculation. From Eqs.\,(\ref{eq:Mxi}) and (\ref{eq:Met}),
we have
\begin{align}
MF & =\sqrt{2}\left(M\xi_{1},...,M\xi_{n},M\eta_{1},...,M\eta_{n}\right)\nonumber \\
 & =\sqrt{2}\left(c_{1}\xi_{1}-s_{1}\eta_{1},...,c_{n}\xi_{n}-s_{n}\eta_{n},s_{1}\xi_{1}+c_{1}\eta_{1},...,s_{n}\xi_{n}+c_{n}\eta_{n}\right)\,.
\end{align}
Then, 
\begin{align}
N & =-JF^{\dagger}JMF\nonumber \\
 & =-2J\left(\begin{array}{c}
\xi_{1}^{\dagger}\\
\vdots\\
\xi_{n}^{\dagger}\\
\eta_{1}^{\dagger}\\
\vdots\\
\eta_{n}^{\dagger}
\end{array}\right)J\left(c_{1}\xi_{1}-s_{1}\eta_{1},...,c_{n}\xi_{n}-s_{n}\eta_{n},s_{1}\xi_{1}+c_{1}\eta_{1},...,s_{n}\xi_{n}+c_{n}\eta_{n}\right)\nonumber \\
 & =-2J\left(\begin{array}{cccccc}
\xi_{1}^{\dagger}J\left(\xi_{1}c_{1}-\eta_{1}s_{1}\right) & \ldots & \xi_{1}^{\dagger}J\left(\xi_{n}c_{n}-\eta_{n}s_{n}\right) & \xi_{1}^{\dagger}J\left(\xi_{1}s_{1}+\eta_{1}c_{1}\right) & \ldots & \xi_{1}^{\dagger}J\left(\xi_{n}s_{n}+\eta_{n}c_{n}\right)\\
\vdots & \ddots & \vdots & \vdots & \ddots & \vdots\\
\xi_{n}^{\dagger}J\left(\xi_{1}c_{1}-\eta_{1}s_{1}\right) & \ldots & \xi_{n}^{\dagger}J\left(\xi_{n}c_{n}-\eta_{n}s_{n}\right) & \xi_{n}^{\dagger}J\left(\xi_{1}s_{1}+\eta_{1}c_{1}\right) & \ldots & \xi_{n}^{\dagger}J\left(\xi_{n}s_{n}+\eta_{n}c_{n}\right)\\
\eta_{1}^{\dagger}J\left(\xi_{1}c_{1}-\eta_{1}s_{1}\right) & \ldots & \eta_{1}^{\dagger}J\left(\xi_{n}c_{n}-\eta_{n}s_{n}\right) & \eta_{1}^{\dagger}J\left(\xi_{1}s_{1}+\eta_{1}c_{1}\right) & \ldots & \eta_{1}^{\dagger}J\left(\xi_{n}s_{n}+\eta_{n}c_{n}\right)\\
\vdots & \ddots & \vdots & \vdots & \ddots & \vdots\\
\eta_{n}^{\dagger}J\left(\xi_{1}c_{1}-\eta_{1}s_{1}\right) & \ldots & \eta_{n}^{\dagger}J\left(\xi_{n}c_{n}-\eta_{n}s_{n}\right) & \eta_{n}^{\dagger}J\left(\xi_{1}s_{1}+\eta_{1}c_{1}\right) & \ldots & \eta_{n}^{\dagger}J\left(\xi_{n}s_{n}+\eta_{n}c_{n}\right)
\end{array}\right)\nonumber \\
 & =-J\left(\begin{array}{cccccc}
-s_{1} & 0 & 0 & c_{1} & 0 & 0\\
0 & \ddots & 0 & 0 & \ddots & 0\\
0 & 0 & -s_{n} & 0 & 0 & c_{n}\\
-c_{1} & 0 & 0 & -s_{1} & 0 & 0\\
0 & \ddots &  & 0 & \ddots & 0\\
0 & 0 & -c_{n} & 0 & 0 & -s_{n}
\end{array}\right)=\left(\begin{array}{cccccc}
c_{1} & 0 & 0 & s_{1} & 0 & 0\\
0 & \ddots & 0 & 0 & \ddots & 0\\
0 & 0 & c_{n} & 0 & 0 & s_{n}\\
-s_{1} & 0 & 0 & c_{1} & 0 & 0\\
0 & \ddots & 0 & 0 & \ddots & 0\\
0 & 0 & -s_{n} & 0 & 0 & c_{n}
\end{array}\right)\nonumber \\
 & =R(\theta_{1})\diamond...\diamond R(\theta_{n})\,,
\end{align}
 where Eqs.\,(\ref{eq:jxi1})-(\ref{eq:jet}) are used again in the
4th equal sign. 
\end{proof}

\section{Horizontal polar decomposition of symplectic matrices \label{sec:Horizontal}}

In this section, we develop the second tool, horizontal polar decomposition
of symplectic matrices, for the purpose of proving Theorem \ref{thm:stable}.
Judging from its name, horizontal polar decomposition bears some resemblance
to the familiar matrix polar decomposition. The adjective ``horizontal''
is of course related to the natures of even-dimensional symplectic
vector spaces. To understand its construction, we start from Lagrangian
subspaces \citep{deGosson06}. 

The standard symplectic matrix $J$ given by Eq.\,(\ref{J}) defines
a 2-form on $\mathbb{R}^{2n}$, 
\begin{equation}
\sigma(z_{1},z_{2})=z_{2}^{\dagger}Jz_{1}\,.
\end{equation}
The space $\mathbb{R}^{2n}$ equipped with $\sigma$ is called the
standard symplectic space $\left(\mathbb{R}^{2n},\sigma\right).$ 

Recall that a subspace $l\subset$$\left(\mathbb{R}^{2n},\sigma\right)$
is called a Lagrangian plane or Lagrangian subspace if it has dimension
$n$ and $\sigma(z_{1},z_{2})=0$ for all $z_{1},z_{2}\in l.$ The
space of all Lagrangian planes is the Lagrangian Grassmannian denoted
by $Lag(2n).$ Two special Lagrangian planes are the horizontal plan
$l_{x}=\mathbb{R}^{n}\times0$ and the vertical plane $l_{p}=0\times\mathbb{R}^{n}$.
Symplectic matrices act on $Lag(2n)$, i.e., for $M\in Sp(2n,\mathbb{R})$
and $l\in Lag(2n),$ $Ml\in Lag(2n).$ The following lemma highlights
the role of the subgroup $Sp(2n,\mathbb{R})\cap O(2n,\mathbb{R})\backsimeq U(n)$
in the action of symplectic group on the Lagrangian Grassmannian. 
\begin{lem}
\label{lem:ulag}The action of $Sp(2n,\mathbb{R})\cap O(2n,\mathbb{R})\backsimeq U(n)$
on $Lag(2n)$ is transitive, i.e., for every pairs of $\left\{ l_{1},l_{2}\right\} \subset Lag(2n),$
there exists a $u\in Sp(2n,\mathbb{R})\cap O(2n,\mathbb{R})\backsimeq U(n)$
such that $l_{2}=ul_{1}.$
\end{lem}

The lemma can be proved by establishing othor-symplectic bases for
$\left(\mathbb{R}^{2n},\sigma\right)$ over $l_{1}$ and $l_{2}$.
See Ref.~\citep{deGosson06} for details. For a given $l\in Lag(2n),$
all $S\in Sp(2n,\mathbb{R})$ that satisfy $Sl=l$ form a subgroup.
It is called the stabilizer or isotropy subgroup of $l,$ and is denoted
by $St(l).$ For the vertical Lagrangian subsapce $l_{p}=0\times\mathbb{R}^{n},$
its stabilizer must be of the form 
\begin{equation}
S=\left(\begin{array}{cc}
L & 0\\
Q & D
\end{array}\right)\,.
\end{equation}
 Furthermore, the symplectic condition requires that $L^{\dagger}Q$
is symmetric and $D=L^{\dagger-1}$, which reduce $S$ to 
\begin{equation}
S=\left(\begin{array}{cc}
L & 0\\
Q & L^{\dagger-1}
\end{array}\right)=\left(\begin{array}{cc}
I & 0\\
P & I
\end{array}\right)\left(\begin{array}{cc}
L & 0\\
0 & L^{\dagger-1}
\end{array}\right)\,,\label{eq:Mst}
\end{equation}
where $P=QL^{-1}$ is symmetric. For any $M\in Sp(2n,\mathbb{R})$,
let's consider the Lagrangian subspaces $M^{-1}l_{p}$ and $l_{p}$.
By Lemma \ref{lem:ulag}, there must exit a $u\in Sp(2n,\mathbb{R})\cap O(2n,\mathbb{R})\backsimeq U(n)$
such that $M^{-1}l_{p}=ul_{p},$ which implies $u^{-1}M^{-1}\in St(l_{p}),$
or, there must exist a $S\in St(l_{p})$ such that $M^{-1}=uS.$ Equivalently,
there must exist a $u^{\prime}\in Sp(2n,\mathbb{R})\cap O(2n,\mathbb{R})\backsimeq U(n)$
and a $S^{\prime}\in St(l_{p})$ such that $M=u^{\prime}S^{\prime}.$
Since $S^{\prime}$ must have the form in Eq.\,(\ref{eq:Mst}), we
conclude that $M$ can always be decomposed as 
\begin{equation}
M=\left(\begin{array}{cc}
X & Y\\
-Y & X
\end{array}\right)\left(\begin{array}{cc}
L & 0\\
Q & L^{\dagger-1}
\end{array}\right)\,,\label{eq:preIw}
\end{equation}
where $L^{\dagger}Q$ is symmetric. Following de Gosson \citep{deGosson06},
this decomposition is called pre-Iwasawa decomposition. The justification
of this terminology will be given shortly. One important fact to realize
is that the pre-Iwasawa decomposition of a symplectic matrix by $Sp(2n,\mathbb{R})\cap O(2n,\mathbb{R})\backsimeq U(n)$
and $St(l_{p})$ is not unique. For any given pre-Iwasawa decomposition
$M=uS$ with $u\in U(n)$ and $S\in St(l_{p}),$ a family of decompositions
can be constructed using $O(n,\mathbb{R})$ as $M=u^{\prime}S^{\prime},$
where 
\begin{align}
u^{\prime} & =ug^{-1}\in Sp(2n,\mathbb{R})\cap O(2n,\mathbb{R})\backsimeq U(n)\,,\\
S^{\prime} & =gS\in St(l_{p})\,,\\
g & =\left(\begin{array}{cc}
c & 0\\
0 & c
\end{array}\right)\,,\,\,\,c\in O(n,\mathbb{R})\,.
\end{align}
This family of pre-Iwasawa decompositions is generated by a $O(n,\mathbb{R})$
gauge freedom. One way to fix the gauge freedom is to demand $L$
to be positive-definite. Different choice of $c\in O(n,\mathbb{R})$
in the family of pre-Iwasawa decomposition corresponds to different
$L$ in Eq.\,(\ref{eq:preIw}). Since the polar decomposition of
a matrix into a positive matrix and a rotation is unique, the pre-Iwasawa
decomposition is unique when the $n\times n$ matrix $L$ is required
to be positive-definite. Because $L$ corresponds to the horizontal
components of $z$, such a requirement demands horizontal positive-definiteness.
We call this special pre-Iwasawa decomposition horizontal polar decomposition
of a symplectic matrix. Vertical polar decomposition can be defined
similarly. 
\begin{thm}
\label{thm:horizontal}(Horizontal polar decomposition) A symplectic
matrix $M\in Sp(2n,\mathbb{R})$ can be uniquely decomposed as
\begin{equation}
M=\left(\begin{array}{cc}
X & Y\\
-Y & X
\end{array}\right)\left(\begin{array}{cc}
L & 0\\
Q & L^{-1}
\end{array}\right)\,,
\end{equation}
where $\left(\begin{array}{cc}
X & Y\\
-Y & X
\end{array}\right)\in Sp(2n,\mathbb{R})\cap O(2n,\mathbb{R})\backsimeq U(n)$, $L$ is positive-definite, and $L^{\dagger}Q$ is symmetric. 
\end{thm}

\begin{proof}
The existence of the decomposition has been established by the pre-Iwasawa
decomposition derived above. We only need to prove the uniqueness.
Assume there are two horizontal polar decompositions, $M=uS=u^{\prime}S^{\prime}.$
We have $SS^{\prime-1}=u^{-1}u^{\prime}\in Sp(2n,\mathbb{R})\cap O(2n,\mathbb{R})\backsimeq U(n),$
or, 
\begin{equation}
\left(\begin{array}{cc}
L & 0\\
Q & L^{-1}
\end{array}\right)\left(\begin{array}{cc}
L^{\prime-1} & 0\\
-Q^{\prime\dagger} & L^{\prime}
\end{array}\right)=\left(\begin{array}{cc}
LL^{\prime-1} & 0\\
QL^{\prime-1}-L^{-1}Q^{\prime\dagger} & L^{-1}L^{\prime}
\end{array}\right)=\left(\begin{array}{cc}
A & B\\
-B & A
\end{array}\right)\,.
\end{equation}
Then, $LL^{\prime-1}=L^{-1}L^{\prime}$, or, $L^{2}=L^{\prime2}.$
Since $L$ and $L^{\prime}$ are positive-definite, $L=L^{\prime}.$
In addition, $QL^{\prime-1}=L^{-1}Q^{\prime\dagger},$ i.e., $LQ=Q^{\prime\dagger}L^{\prime},$
which gives $Q=Q^{\prime},$ considering the symplectic condition
$Q^{\prime\dagger}L^{\prime}=L^{\prime}Q^{\prime}.$ 
\end{proof}
Similar to the standard polar decomposition of an invertible square
matrix, there is an explicit formula for the horizontal polar decomposition
of a symplectic matrix in terms of its block components. 
\begin{cor}
The horizontal polar decomposition of a symplectic matrix 
\begin{equation}
M=\left(\begin{array}{cc}
A & B\\
C & D
\end{array}\right)
\end{equation}
is given by 
\begin{align}
M & =\left(\begin{array}{cc}
X & Y\\
-Y & X
\end{array}\right)\left(\begin{array}{cc}
L & 0\\
Q & L^{-1}
\end{array}\right)\,,\label{eq:Shor}\\
X & =DL\,,\,\,\,Y=BL\,,\label{eq:ShorXY}\\
L & =\left(B^{\dagger}B+D^{\dagger}D\right)^{-1/2}\,,\label{eq:ShorL}\\
Q & =L(B^{\dagger}A+D^{\dagger}C)\,.\label{eq:ShorQ}
\end{align}
\end{cor}

\begin{proof}
Carrying out the matrix multiplication on the right-hand-side of Eq.\,(\ref{eq:Shor}),
we have, block by block, 
\begin{align}
A & =XL+YQ\,,\label{eq:ShorA}\\
C & =-YL+XQ\label{eq:ShorC}\\
B & =YL^{-1}\,,\,\,\,D=XL^{-1}\,.\label{eq:ShorBD}
\end{align}
From Eq.\,(\ref{eq:ShorBD}), we have Eq.\,(\ref{eq:ShorXY}) and
\begin{equation}
B^{\dagger}B+D^{\dagger}D=L^{-1}(X^{\dagger}X+Y^{\dagger}Y)L^{-1}=L^{-2},
\end{equation}
 where use is made of $X^{\dagger}X+Y^{\dagger}Y=I.$ Since $L$ is
positive-definite, it's square-root is unique, and Eq.\,(\ref{eq:ShorL})
follows. Equations (\ref{eq:ShorXY}) and (\ref{eq:ShorL}) confirm
that $\left(\begin{array}{cc}
X & Y\\
-Y & X
\end{array}\right)$ indeed belongs to $U(n).$ Summing $X^{\dagger}$Eq.\,(\ref{eq:ShorC})
and $Y^{\dagger}$Eq.\,(\ref{eq:ShorA}) gives 
\begin{equation}
X^{\dagger}C+Y^{\dagger}A=\left(Y^{\dagger}X-X^{\dagger}Y\right)L+\left(Y^{\dagger}Y+X^{\dagger}X\right)Q=Q\,,
\end{equation}
which proves Eq.\,(\ref{eq:ShorQ}). It is straightforward to verify
that Eqs.\,(\ref{eq:ShorA}) and (\ref{eq:ShorC}) hold when $X$,
$Y,$ $L,$ $Q$ are specified by Eqs.\,(\ref{eq:ShorXY})-(\ref{eq:ShorQ})
.
\end{proof}
Horizontal polar decomposition is not the only choice to fix the $O(n,\mathbb{R})$
gauge freedom in the pre-Iwasawa decomposition. Another possibility
is to require $L$ assume the form of 
\begin{equation}
L=\left(\begin{array}{cccc}
e^{\alpha_{1}} & l_{12} & \ldots & l_{1n}\\
0 & e^{\alpha_{2}} & \ldots & l_{2n}\\
0 & 0 & \ddots & \vdots\\
0 & 0 & 0 & e^{\alpha_{n}}
\end{array}\right)\,.
\end{equation}
Under this restriction, the decomposition is unique and this is the
well-known Iwasawa decomposition for symplectic matrices. The Iwasawa
decomposition is a general result for all Lie groups \citep{Iwasawa49}.
However, it is not found relevant to the present study, except that
it justify the terminology of pre-Iwasawa decomposition in the form
of Eq.\,(\ref{eq:preIw}) \citep{deGosson06}. 

\section{Proof of Theorem 2 \label{sec:Proof 2}}

We are now ready to prove Theorem \ref{thm:stable}, by invoking the
normal forms for stable symplectic matrices and horizontal polar decomposition
for symplectic matrices. 

The sufficiency of Theorem \ref{thm:stable} follows from the expression
of the solution map given by Eq.\,(\ref{M}) for system (\ref{zdot})
in Theorem \ref{thm:sol}, and the associated gauge freedom. For the
matrix $S$ at $t=0,$ consider the gauge transformation induced by
a $c_{0}\in O(n,\mathbb{R}),$
\begin{align}
\tilde{S}_{0} & =g_{0}S_{0}=\left(\begin{array}{cc}
\tilde{w}_{o}^{\dagger-1} & 0\\
(\tilde{w}_{0}R-\dot{\tilde{w}}_{0})m & \tilde{w}_{0}
\end{array}\right)\,,\\
g_{0} & =\left(\begin{array}{cc}
c_{0} & 0\\
0 & c_{0}
\end{array}\right)\in Sp(2n,\mathbb{R})\cap O(2n,\mathbb{R})\backsimeq U(n)\,,\\
\tilde{w}_{0} & =c_{0}w_{0}\,,\,\,\,w_{0}=w(t=0)\,.
\end{align}
Specifically, we select $c_{0}=\sqrt{w_{0}^{\dagger}w_{0}}w_{0}^{-1}$
such that $\tilde{w}_{0}=\sqrt{w_{0}^{\dagger}w_{0}}$ is positive-definite.
Clearly, the gauge transformation amounts to the polar decomposition
of $w_{0}$ and makes $\tilde{S}_{0}$ horizontally positive-definite.
A similar gauge transformation is applied at $t=T$ using $c_{T}=\sqrt{w_{T}^{\dagger}w_{T}}w_{T}^{-1}$,
\begin{align}
\tilde{S}_{T} & =g_{T}S_{T}=\left(\begin{array}{cc}
\tilde{w}_{T}^{\dagger-1} & 0\\
(\tilde{w}_{T}R-\dot{\tilde{w}}_{T})m & \tilde{w}_{T}
\end{array}\right)\,,\label{eq:STt}\\
g_{T} & =\left(\begin{array}{cc}
c_{T} & 0\\
0 & c_{T}
\end{array}\right)\in Sp(2n,\mathbb{R})\cap O(2n,\mathbb{R})\backsimeq U(n)\,,\\
\tilde{w}_{T} & =c_{T}w_{T}\,,\,\,\,w_{T}=w(t=T)\,.
\end{align}
 With these two gauge transformations at $t=0$ and $T$, the one-period
solution map for system (\ref{zdot}) is 
\begin{equation}
M(T)=\tilde{S}_{T}^{-1}g_{T}P_{T}g_{0}^{-1}\tilde{S}_{0}\,.
\end{equation}
If the envelope equation (\ref{w}) admits a solution $w(t)$ such
that $w_{T}^{\dagger}w_{T}=w_{0}^{\dagger}w_{0}$ and $S_{0}$ is
symplectic, then $M(T)$ is stable because 
\begin{equation}
M^{l}(T)=\tilde{S}_{0}^{-1}\left(g_{T}P_{T}g_{0}^{-1}\right)^{l}\tilde{S}_{0}\,
\end{equation}
 and $g_{T}P_{T}g_{0}^{-1}\in Sp(2n,\mathbb{R})\cap O(2n,\mathbb{R})\backsimeq U(n).$
The sufficiency of Theorem \ref{thm:stable} is proved.

For necessity, assume $M(T)$ is stable. By Theorem \ref{thm:normal},
$M(T)$ can be written as 
\begin{equation}
M(T)=F^{-1}NF\,,
\end{equation}
where $N=R(\theta_{1})\diamond R(\theta_{2})...\diamond R(\theta_{n})\in Sp(2n,\mathbb{R})\cap O(2n,\mathbb{R})\backsimeq U(n)$
and $F\in Sp(2n,\mathbb{R})$. Let the horizontal polar decomposition
of $F$ is 
\begin{equation}
F=P_{F}S_{F}\,.
\end{equation}
 Thus,
\begin{equation}
M(T)=S_{F}^{-1}P_{F}^{-1}NP_{F}S_{F}\,.\label{eq:MTSP}
\end{equation}
We choose initial conditions $w_{0}$ and $\dot{w}_{0}$ such that
\begin{equation}
\tilde{S}_{0}=S_{0}=\left(\begin{array}{cc}
w_{o}^{\dagger-1} & 0\\
(w_{0}R-\dot{w}_{0})m & w_{0}
\end{array}\right)=S_{F}\,,
\end{equation}
which can always be accomplished as $m$ is invertible. Note that
both $S_{F}$ and $S_{0}$ belong to $St(l_{p})$ and are horizontally
positive-definite. This choice of initial conditions $w_{0}$ and
$\dot{w}_{0}$ uniquely determines the dynamics of the envelope matrix
$w(t)$ as well as the solution map, 
\begin{equation}
M(T)=\tilde{S}_{T}^{-1}g_{T}P_{T}S_{0}=\tilde{S}_{T}^{-1}g_{T}P_{T}S_{F}\,.\label{eq:MTSP2}
\end{equation}
Equations (\ref{eq:MTSP}) and (\ref{eq:MTSP2}) show that 
\begin{equation}
\tilde{S}_{T}^{-1}g_{T}P_{T}=S_{F}^{-1}P_{F}^{-1}NP_{F}\,,
\end{equation}
or
\begin{equation}
(g_{T}P_{T})^{-1}\tilde{S}_{T}=(P_{F}^{-1}NP_{F})^{-1}S_{F}\,.
\end{equation}
By construction in Eq.\,(\ref{eq:STt}), $\tilde{S}_{T}$ belongs
to $St(l_{p})$ and is horizontally positive-definite. In addition,
both $(g_{T}P_{T})^{-1}$ and $(P_{F}^{-1}NP_{F})^{-1}$ are in $Sp(2n,\mathbb{R})\cap O(2n,\mathbb{R})\backsimeq U(n).$
By the uniqueness of horizontal polar decomposition, 
\begin{equation}
\tilde{S}_{T}=S_{F}=S_{0}=\tilde{S}_{0}\,.
\end{equation}
Therefore, the envelope matrix $w(t)$ determined by the initial conditions
$w_{0}$ and $\dot{w}_{0}$ satisfy the requirement that $\sqrt{w^{\dagger}w}$
is periodic with periodicity $T$ and $S_{0}$ is symplectic. 

This completes the proof of Theorem \ref{thm:stable}.

\section{Conclusions and future work \label{sec:Conclusions}}

Linear Hamiltonian systems with periodic coefficients have many important
applications in physics and nonlinear Hamiltonian dynamics. One of
the key issues is the stability of the systems. In this paper, we
have established in Theorem \ref{thm:stable} a necessary and sufficient
condition for the stability of the systems, in terms of solutions
of an associated matrix envelope equation. The envelope matrix governed
by the envelope equation plays a central role in determining the dynamic
properties of the linear Hamiltonian systems. Specifically, the envelope
matrix is the most important building block of the solution map given
by Theorem \ref{thm:sol}; it encapsulates the slow dynamics of the
envelope of the fast oscillation, when the dynamics has a time-scale
separation; it also controls how the fast dynamics evolves. 

Three tools are utilized in the study. The method of time-dependent
canonical transformation is used to construct the solution map and
derive the envelope equation. The normal forms for stable symplectic
matrices (Theorem \ref{thm:normal}) and horizontal polar decomposition
for symplectic matrices (Theorem \ref{thm:horizontal}) are developed
to prove the necessary and sufficient condition for stability. These
tools systematically decompose the dynamics of linear Hamiltonian
systems with time-dependent coefficients, and are expected to be effective
in other studies as well, such as those on quantum algorithms for
classical Hamiltonian systems. Relevant results will be reported in
future publications. 
\begin{acknowledgments}
I would like to acknowledge fruitful discussions on the subject studied
in this paper with collaborators, colleagues, and friends, including
Joshua Burby, Moses Chung, Robert Dewar, Nathaniel Fisch, Vasily Gelfreich,
Alexander Glasser, Maurice de Gosson, Oleg Kirillov, Melvin Leok,
Yiming Long, Robert MacKay, Richard Montgomery, Philip Morrison, Yuan
Shi, Jianyuan Xiao, Ruili Zhang, and Chaofeng Zhu. Especially, I would
like to thank Profs. Yiming Long, Chaofeng Zhu, and Richard Montgomery
for detailed discussions on normal forms for symplectic matrices,
and Prof. Maurice de Gosson for detailed discussion on pre-Iwasawa
decomposition and horizontal polar decomposition. This paper was presented
in a Lunch with Hamilton Seminar on September 12, 2018 at the Mathematical
Sciences Research Institute, as a part of the Program on Hamiltonian
systems, from topology to applications through analysis. I would like
to thank Prof. Philip Morrison for the invitation and Prof. Amitava
Bhattacharjee for the support to participate in the Program. This
research was supported by the U.S. Department of Energy (DE-AC02-09CH11466). 

Finally, I would like to dedicate this paper to the late Prof. Ronald
C. Davidson. 
\end{acknowledgments}

\bibliographystyle{apsrev4-1}
\bibliography{lh}

\end{document}